\documentclass[11pt]{article}

\usepackage{amsmath,amsthm,amsfonts,amssymb}

\usepackage{algorithm,algorithmic}

\usepackage{pgfplots}
\pgfplotsset{compat=1.14}

\newtheorem{lemma}{Lemma}
\newtheorem{proposition}{Proposition}
\newtheorem{corollary}{Corollary}

\usepackage[a4paper,left=2.5cm,right=2.5cm]{geometry}

\begin{document}

\title{An Improved Exact Sampling Algorithm for the Standard Normal Distribution\thanks{This work was supported by National Key R\&D Program of China (2017YFB0802500), the Major Program of Guangdong Basic and Applied Research (2019B030302008), National Natural Science Foundations of China (Grant Nos.\,61672550, 61972431), and the Fundamental Research Funds for the Central Universities (Grant No.\,19lgpy217).}}

\author{
Yusong Du, Baoying Fan, and Baodian Wei\\
\normalsize{School of Data and Computer Science, Sun Yat-sen University, Guangzhou 510006, China}\\
\normalsize{Guangdong Key Laboratory of Information Security Technology, Guangzhou 510006, China}\\
\small{duyusong@mail.sysu.edu.cn}
}
\date{}
\maketitle

\begin{abstract}
In 2016, Karney proposed an exact sampling algorithm for the standard normal distribution. In this paper, we study the computational complexity of this algorithm under the random deviate model. Specifically, Karney's algorithm requires the access to an infinite sequence of independently and uniformly random deviates over the range $(0,1)$. We give an estimate of the expected number of uniform deviates used by this algorithm until outputting a sample value, and present an improved algorithm with lower uniform deviate consumption. The experimental results also shows that our improved algorithm has better performance than Karney's algorithm.
\end{abstract}

\newpage

\section{Introduction}
We denote the set of real numbers by $\mathbb{R}$. The Gaussian function on $\mathbb{R}$ with parameter $\sigma>0$ and $\mu\in\mathbb{R}$ evaluated at $x\in\mathbb{R}$ can be defined by $$\rho_{\sigma,\mu}(x)=\exp{\left(-\frac{(x-\mu)^2}{2\sigma^2}\right)}.$$
Normalizing $\rho_{\sigma,\mu}(x)$ by its total measure $\int_{x\in\mathbb{R}}\rho_{\sigma,\mu}(x)=\sqrt{2\pi}\sigma$ over $\mathbb{R}$, we obtain the probability density function of the (continuous) Gaussian distribution $\mathcal{N}(\mu,\sigma)$, namely the normal distribution of mean $\mu$ and variance $\sigma^2$.

In 2016, Karney \cite{Karney2016} proposed an exact sampling algorithm for the standard normal distribution $\mathcal{N}(0,1)$ of mean $0$ and variance $1$. This algorithm uses rejection sampling \cite{Devroye1986d}, requires no floating-point arithmetic, and generates sample values that conform to $\mathcal{N}(0,1)$ without any statistical discrepancy if we have a source of perfectly uniform deviates over the range $(0,1)$ (essentially, uniformly random bits) at our disposal. Although this algorithm is unlikely to displace existing methods for most applications, it is of theoretical interest as an example of an algorithm in which exact transcendental results can be achieved with simple integer arithmetic. 

Karney's exact sampling algorithm for the standard normal distribution involves exactly sampling from a discrete Gaussian distribution of mean $\mu=0$ and variance $1$ over the set of all the non-negative integers. We denote this distribution by $\mathcal{D}_{\mathbb{Z}^+,1}$, where $\mathbb{Z}^+$ is the set of non-negative integers. This distribution is simply the standard normal distribution restricted so that its support is $\mathbb{Z}^+$. Generally, for all $k\in\mathbb{Z}^+$, a discrete Gaussian distribution of mean $\mu$ and variance $\sigma^2$ over the set of non-negative integers $\mathbb{Z}^+$ can be defined by
$$\mathcal{D}_{\mathbb{Z}^+,\sigma,\mu}(k)={\rho_{\sigma,\mu}(k)}/{\rho_{\sigma,\mu}(\mathbb{Z}^+)},$$
where $\rho_{\sigma,\mu}(\mathbb{Z}^+)=\sum_{x\in\mathbb{Z}^+}\rho_{\sigma,\mu}(x)$. By convention, the subscript $\mu$ is omitted when it is taken to be $0$. Karney's algorithm also involves exactly sampling from the Bernoulli distribution
$$\mathcal{B}_{\exp\left(-\frac12x(2k+x)\right)}$$
with an integer $k\in\mathbb{Z}^+$ and a real number $x\in(0,1)$. This is equivalent to generating a Bernoulli random value (true or false) which is true with probability $\exp\left(-\frac12x(2k+x)\right)$. 

\subsection{Our Contribution}
In this paper, we study the computational complexity of Karney's exact sampling algorithm for the standard normal distribution under the random deviate model. In this model, we have the access to an infinite sequence of independently and uniformly random deviates over the range $(0,1)$, and the complexity is measured by the expected number of uniform deviates used until the algorithm outputting a sample value. Specifically, we give an estimate of the expected number of uniform deviates used by Karney's algorithm, and present an improved algorithm with lower uniform deviate consumption. The experimental results also shows that our improved algorithm has better actual performance than Karney's algorithm.

The improved sampling algorithm comes from our two observations on Karney's sampling algorithm for the standard normal distribution. We find that the sampling method for the discrete Gaussian distribution $\mathcal{D}_{\mathbb{Z}^+,1}(k)={\rho_{1}(k)}/{\rho_{1}(\mathbb{Z}^+)}$ in Karney's algorithm is not the optimal. It uses about $4.8265$ Bernoulli random values from $\mathcal{B}_{1/\sqrt{e}}$ on average (see Proposition \ref{ComplexityStd}), and each Bernoulli random value from $\mathcal{B}_{1/\sqrt{e}}$ requires $\sqrt{e}$ uniform deviates on average. We will give an improved algorithm with lower uniform deviate consumption for sampling $\mathcal{D}_{\mathbb{Z}^+,1}$, which only requires about $3.684$ Bernoulli random values from $\mathcal{B}_{1/\sqrt{e}}$ on average to output an integer from $\mathcal{D}_{\mathbb{Z}^+,1}$ (see Proposition \ref{ComplexityStdNew}). We also find that the suggested way of sampling from the Bernoulli distribution $\mathcal{B}_{\exp\left(-\frac12x(2k+x)\right)}$ in Karney's algorithm consumes about $2.194$ uniform deviates on average. This is a relatively large number. We will present an alternative algorithm so that one may avoid generating the Bernoulli value in the suggested way and reduce the expected number of uniform deviates used to $2.018$. 

\subsection{The Expected Number of Uniform Deviates}
In fact, as a rejection sampling algorithm, its acceptance/rejection rate can be regarded as the computational complexity. In Karney's paper \cite{Karney2016}, it has been indicated that the rejection rate of the sampling algorithm for the standard normal distribution is $(\sqrt{2/\pi})/(1-1/\sqrt{e})\approx2.03$, which means that the algorithm executes about $2.03$ times on average before it outputs a sample value. However, the rejection rate is only a very ``coarse-grained'' complexity model, which could not accurately reflect the actual computational cost. Therefore, in order to analyze the computational complexity of Karney's algorithm more accurately, we consider the expected number of uniform deviates used, namely, the random deviate model, which can not only cover the computational complexity of sampling the proposal distribution $\mathcal{D}_{\mathbb{Z}^+,1}$ but also reflect the computational cost required by the rejection operation. The impact of rejection rate on the complexity of the sampling algorithm can also be fully reflected in the random deviate model.

The computational complexity of von Neumann's exact sampling algorithm for the exponential distribution $e^{-x}$ with $x\geq0$ was given also by estimating the expected number of uniform deviates used \cite{vonNeumann1951}. Karney improved von Neumann’s algorithm by introducing the early rejection step, which reduces the expected number of uniform deviates used from $e/(1-e^{-1})\approx4.30$ to $\sqrt{e}/(1-1/\sqrt{e})\approx4.19$. (see Algorithm E in \cite{Karney2016}.) Our work in this paper is of a similar nature. We study the computational complexity by analyzing and reducing the expected number of uniform deviates used by Karney's exact sampling algorithm for the standard normal distribution, as we note that there is no theoretical estimate of the expected number of uniform deviates for this algorithm.

Another common complexity model for random sampling algorithms is the random bit model, in which the complexity is measured by the expected number of uniformly random bits used by the sampling algorithm \cite{KY76,Flajolet1986,Devroye2017}. Karney also indicated that the expected number of random bits consumed by his algorithm for the standard normal distribution is about $30.0$, but this was only an empirical value obtained by experiment \cite{Karney2016}. A uniformly random deviate is always made up of an indefinite number of uniformly random bits, so the random bit model is a more ``fine-grained'' complexity model as compared to the random deviate model. Nonetheless, we will only use the random deviate model to analyze the the computational complexity in this work, because the random source required by Karney's algorithm is mainly in the form of uniform deviates over the range $(0,1)$. The random deviate model could help us understand the complexity of Karney's algorithm more intuitively. Furthermore, the random deviate model is simpler than the random bit model, and it can make the whole process of complexity analysis more concise.

\section{Karney's Sampling Algorithm}
Algorithm~\ref{KarneyStdN} is Karney's algorithm for sampling exactly from the standard (continuous) normal distribution $\mathcal{N}(0,1)$ of mean $0$ and variance $1$.

\begin{algorithm}
\caption{\cite{Karney2016} Sampling from the standard normal distribution $\mathcal{N}(0,1)$}\label{KarneyStdN}
\begin{algorithmic}[1]
\ENSURE a sample value from $\mathcal{N}(0,1)$
\STATE select $k\in\mathbb{Z}^{+}$ with probability $\exp(-k/2)\cdot(1-\exp(-1/2))$.
\STATE accept $k$ with probability $\exp\left(-\frac{1}{2}k(k-1)\right)$, otherwise \textbf{goto} step 1.
\STATE sample a uniformly random number $x\in(0,1)$.
\STATE accept $x$ with probability $\exp\left(-\frac12x(2k+x)\right)$, otherwise \textbf{goto} step 1.
\STATE set $s\leftarrow \pm1$ with equal probabilities and \textbf{return} $s(k+x)$.
\end{algorithmic}
\end{algorithm}

From the perspective of rejection sampling, in Algorithm~\ref{KarneyStdN}, we can see that step 1 and step 2 form a rejection sampling procedure, which samples $k\in\mathbb{Z}^+$ from the discrete Gaussian distribution $\mathcal{D}_{\mathbb{Z}^+,1}={\rho_{1}(k)}/{\rho_{1}(\mathbb{Z}^+)}$. Step 4 is equivalent to sampling from the Bernoulli distribution $\mathcal{B}_{\exp\left(-\frac12x(2k+x)\right)}$. The correctness can seen from the fact that
$$\rho_{1}(k)\cdot\exp\left(-\frac12x(2k+x)\right)=\exp{\left(-\frac{(k+x)^2}{2}\right)}.$$

In Algorithm~\ref{KarneyStdN}, in order to exactly sample $k\in\mathbb{Z}^{+}$ with relative probability density $\rho_{1}(k)=\exp({-k^2}/{2})$, one needs to exactly sample Bernoulli random values according to $\mathcal{B}_{1/\sqrt{e}}$ by applying Algorithm~\ref{BernoulliExph}. Specifically, one applies Algorithm~\ref{BernoulliExph} repeatedly $(k+1)$ times to select an integer $k\geq0$ with probability $\exp(-k/2)\cdot(1-\exp(-1/2))$ (step 1 in Algorithm~\ref{KarneyStdN}), then continues to apply Algorithm~\ref{BernoulliExph} at most $k(k-1)$ times to accept or reject $k$ with probability $\exp\left(-\frac{1}{2}k(k-1)\right)$ (step 2 in Algorithm~\ref{KarneyStdN}).

\begin{algorithm}
\caption{\cite{Karney2016} Generating a Bernoulli random value which is true with probability $1/\sqrt{e}$}\label{BernoulliExph}
\begin{algorithmic}[1]
\ENSURE a Boolean value according to $\mathcal{B}_{1/\sqrt{e}}$
\STATE sample uniform deviates $u_1, u_2, \ldots$ with $u_i\in(0,1)$ and determine the maximum value $n\geq0$ such that $1/2>u_1>u_2>\ldots>u_n$
\STATE \textbf{return} \textbf{true} if $n$ is even, otherwise \textbf{return} \textbf{false} if $n$ is odd.
\end{algorithmic}
\end{algorithm}

Algorithm~\ref{BernoulliExph} is adapted from Von Neumann's algorithm \cite{vonNeumann1951} for exactly sampling from the exponential distribution $e^{-x}$ for real $x > 0$. More precisely, the probability that the length of the longest decreasing sequence is $n$ is ${x^n}/{n!}-{x^{n+1}}/{(n+1)!}$, and the probability that $n$ is even is exactly equal to
$$(1-x)+\left(\frac{x^2}{2!}-\frac{x^3}{3!}\right)+\ldots+=\sum_{n=0}^{\infty}\left(\frac{x^n}{n!}-\frac{x^{n+1}}{(n+1)!}\right)=e^{-x}.$$

In Algorithm~\ref{KarneyStdN}, step 4 will also be implemented by using a specifically designed algorithm, so that it produces no statistical discrepancy. The main idea is to sample two sets of uniform deviates $u_1, u_2, \ldots$ and $v_1, v_2, \ldots$ from $(0,1)$, and to determine the maximum value $n\geq 0$ such that $x>u_1>u_2>\ldots>u_n$ and $v_i<(2k+x)/(2k+2)$. If $n$ is even, it returns \textbf{true}, and the probability is exactly equal to
$$1-x\left(\frac{2k+x}{2k+2}\right)+\frac{x^2}{2!}\left(\frac{2k+x}{2k+2}\right)^2-\frac{x^3}{3!}\left(\frac{2k+x}{2k+2}\right)^3+\ldots=\exp\left(-x\frac{2k+x}{2k+2}\right).$$
This procedure can be summarized as Algorithm~\ref{BernoulliExpX}.
\begin{algorithm}
\caption{ \cite{Karney2016} Generating a Bernoulli random value which is true with probability $\exp(x\frac{2k+x}{2k+2})$ with an integer $k>0$ and a real number $x\in(0,1)$}
\label{BernoulliExpX}
\begin{algorithmic}[1]
\ENSURE a Boolean value according to $\exp(-x\frac{2k+x}{2k+2})$
\STATE set $y\leftarrow x$, $n\leftarrow 0$.
\STATE sample uniform deviates $z$ with $z\in(0,1)$, and go to step 6 unless $z<y$.
\STATE set $f\leftarrow C(2k+2)$; if $f<0$ go to step 6.
\STATE sample uniform deviates $r\in(0,1)$ if $f=0$, and go to step 6 unless $r<x$.
\STATE set $y\leftarrow z$, $n\leftarrow n+1$; \textbf{goto} step 2.
\STATE \textbf{return} \textbf{true} if $n$ is even, otherwise \textbf{return} \textbf{false} if $n$ is odd.
\end{algorithmic}
\end{algorithm}

In Algorithm~\ref{BernoulliExpX}, $C(m)$ with $m=2k+2$ is a random selector, which returns $-1$, $0$ and $1$ with probability $1/m$, $1/m$ and $1-2/m$ respectively. By applying Algorithm~\ref{BernoulliExpX} at most $k+1$ times, since
$$\exp\left(-\frac12x(2k+x)\right)=\exp\left(-x\frac{2k+x}{2k+2}\right)^{k+1},$$
one can obtain a Bernoulli random value which is true with probability 
$$\exp\left(-\frac12x(2k+x)\right)$$ 
for given $k$ and $x$.

\section{The Computational Complexity of Sampling from $\mathcal{D}_{\mathbb{Z}^+,1}$}
In this section, we analyze the expected number of uniform deviates used by step 1 and step 2 in Algorithm~\ref{KarneyStdN}, which generate a non-negative integer according to $\mathcal{D}_{\mathbb{Z}^+,1}$. Then, we give an improved sampling algorithm with lower uniform deviate consumption.

\subsection{The Complexity of Karney's Algorithm for Sampling from $\mathcal{D}_{\mathbb{Z}^+,1}$}
\begin{proposition}
Algorithm~\ref{BernoulliExph} uses $e^{1/2}$ uniform deviates on average to output a Bernoulli value according to $\mathcal{B}_{1/\sqrt{e}}$.\label{ComplexityBernoulliExph}
\end{proposition}
\begin{proof}
Let $x=1/2$. Algorithm~\ref{BernoulliExph} uses a uniform deviate for each comparison, and the probability that the length of the longest decreasing sequence is $n$ is ${x^n}/{n!}-{x^{n+1}}/{(n+1)!}$. Thus, the expected number of uniform deviates used can be given by
$$(1-x)\cdot1+\left(\frac{x^1}{1!}-\frac{x^2}{2!}\right)\cdot2+\left(\frac{x^2}{2!}-\frac{x^3}{3!}\right)\cdot3+\ldots+=\sum_{n=1}^{\infty}\left(\frac{x^{n-1}}{(n-1)!}-\frac{x^n}{n!}\right)\cdot{n}=e^{x}.$$
Algorithm~\ref{BernoulliExph} uses $e^{x}=e^{1/2}$ uniform deviates on average to output a Bernoulli value.
\end{proof}

Let $x$ be an arbitrary number over the range $(0,1)$. Replacing $1/2$ with $x$ in Algorithm~\ref{BernoulliExph}, we can get a Bernoulli random value which is true with probability $e^{-x}$, namely sampling from $\mathcal{B}_{e^{-x}}$ with $x\in(0,1)$.  Specifically, we sample uniform deviates $u_1, u_2, \ldots$ with $u_i\in(0,1)$ and determine the maximum value $n\geq0$ such that $x>u_1>u_2>\ldots>u_n$, then we get \textbf{true} if $n$ is even, or \textbf{false} if $n$ is odd. This can be viewed as a generalized version of Algorithm~\ref{BernoulliExph}, which has been used implicitly in Karney's algorithms. It is clear that Proposition~\ref{ComplexityBernoulliExph} also holds for any $x\in(0,1)$.

\begin{corollary}
Let $x\in(0,1)$. Replacing $1/2$ with $x$ in Algorithm~\ref{BernoulliExph}, we can use $e^x$ uniform deviates on average to obtain a Bernoulli value which is true with probability $e^{-x}$. \label{ComplexityBernoulliExpx}
\end{corollary}

Since the expected number of uniform deviates used by Algorithm~\ref{BernoulliExph} is a constant, it suffices to consider the expected number of Bernoulli random values from $\mathcal{B}_{1/\sqrt{e}}$ that are required for step 1 and step 2 in Algorithm~\ref{KarneyStdN}.

\begin{proposition}
The expected number of Bernoulli random values from $\mathcal{B}_{1/\sqrt{e}}$ used by step 1 and step 2 for sampling from $\mathcal{D}_{\mathbb{Z}^+,1}$ in Algorithm~\ref{KarneyStdN} is about $4.8265$.\label{ComplexityStd}
\end{proposition}
\begin{proof}
Let $p_1=1/\sqrt{e}$ and $p_0=1-p_1$. Assume that $k\geq0$ is an integer which is generated in step 1. Then, the expected number of Bernoulli random values from $\mathcal{B}_{1/\sqrt{e}}$ used by Algorithm~\ref{KarneyStdN} is
$$\sum_{k=0}^\infty(k+1+k(k-1))p_1^kp_0p_1^{k(k-1)}=\sum_{k=0}^\infty(1+k^2)p_1^{k^2}p_0,$$
if $k$ is accepted in step 2, and the expected number of Bernoulli random values is
$$\sum_{k=2}^\infty p_1^kp_0\left(\sum_{j=1}^{k(k-1)}j\cdot p_1^{j-1}\cdot p_0\right),$$
if $k$ is rejected in step 2. In particular, if $k=0$ or $k=1$, then $k$ is directly accepted in step 2. Therefore, the the expected number of Bernoulli random values from $\mathcal{B}_{1/\sqrt{e}}$ used by Algorithm~\ref{KarneyStdN} for executing step 1 and step 2 at a time can be given by
$$\sum_{k=0}^\infty(1+k^2)p_0p_1^k+\sum_{k=2}^\infty p_1^kp_0\left(\sum_{j=1}^{k(k-1)}(k+1+j)\cdot p_1^{j-1}p_0\right)\approx3.32967.$$
Furthermore, it not hard to see that the probability of Algorithm~\ref{KarneyStdN} not going back to step 1 in step 2 is $\sum_{k=0}^\infty p_1^{k^2}p_0\approx0.689875$. So, the expected number of Bernoulli random values from $\mathcal{B}_{1/\sqrt{e}}$ used by Algorithm~\ref{KarneyStdN} for successfully generating an integer from $\mathcal{D}_{\mathbb{Z}^+,1}$ is about $3.32967/0.689875\approx4.82649$.
\end{proof}

\subsection{An Improved Sampling Algorithm for $\mathcal{D}_{\mathbb{Z}^+,1}$}
Karney's algorithm for sampling from $\mathcal{D}_{\mathbb{Z}^+,1}$ can be easily extended to the case of discrete Gaussian distribution $\mathcal{D}_{\mathbb{Z},\sigma}$ with a rational-valued $\sigma>\sqrt{2}/2$. So, we obtain Algorithm~\ref{samplingCentered}.
\begin{algorithm}
\caption{Sampling from $\mathcal{D}_{\mathbb{Z}^+,\sigma}$ with a rational-valued $\sigma>\sqrt{2}/2$}\label{samplingCentered}
\begin{algorithmic}[1]
\ENSURE an integer $k$ that conforms to $\mathcal{D}_{\mathbb{Z}^+,\sigma}$
\STATE select $k\in\mathbb{Z}^{+}$ with probability $\exp(-k/(2\sigma^2))\cdot(1-\exp(-1/(2\sigma^2))$.
\RETURN $k$ with probability $\exp(-\frac{1}{2\sigma^2}k(k-1))$.
\end{algorithmic}
\end{algorithm}

It is not hard to see that Algorithm~\ref{samplingCentered} relies on exactly sampling from the Bernoulli distribution $\mathcal{B}_{\exp(-1/(2\sigma^2))}$. Since $0<1/2\sigma^2<1$, as mentioned in Section 3.1, by replacing $1/2$ with $1/2\sigma^2$ in Algorithm~\ref{BernoulliExph}, one can get a Bernoulli random value which is true with probability $\exp(-1/(2\sigma^2))$. By Corollary~\ref{ComplexityBernoulliExpx}, it uses $\exp(1/(2\sigma^2))$ uniform deviates on average to obtain a Bernoulli value from $\mathcal{B}_{\exp(-1/(2\sigma^2))}$.

Following the proof of Proposition~\ref{ComplexityStd} with $p_1=\exp(-1/(2\sigma^2))$ and $p_0=1-p_1$, one can easily estimate the expected number of Bernoulli random values from $\mathcal{B}_{\exp(-1/(2\sigma^2))}$ used by Algorithm~\ref{samplingCentered}. Here, we give an improved version of Algorithm~\ref{samplingCentered}, namely Algorithm~\ref{samplingCenteredNew}. It uses a smaller number of Bernoulli random values from $\mathcal{B}_{\exp(-1/(2\sigma^2))}$, and thus has lower computational complexity compared to Algorithm~\ref{samplingCentered}.
\begin{algorithm}
\caption{Sampling from $\mathcal{D}_{\mathbb{Z},\sigma}$ with a rational-valued $\sigma>\sqrt{2}/2$}\label{samplingCenteredNew}
\begin{algorithmic}[1]
\ENSURE an integer $k$ that conforms to $\mathcal{D}_{\mathbb{Z},\sigma}$
\STATE sample $b\leftarrow \mathcal{B}_{\exp(-1/(2\sigma^2))}$ and \textbf{return} $0$ \textbf{if} $b$ is false.
\STATE sample $b\leftarrow \mathcal{B}_{\exp(-1/(2\sigma^2))}$ and \textbf{return} $1$ \textbf{if} $b$ is false.
\STATE set $k\leftarrow 2$ and $t\leftarrow 2(k-1)$.
\STATE \textbf{while} $t>0$ \textbf{do}
\STATE \hspace*{12pt} sample $b\leftarrow \mathcal{B}_{\exp(-1/(2\sigma^2))}$.
\STATE \hspace*{12pt} \textbf{if} $b$ is false, then \textbf{goto} step 1, otherwise set $t\leftarrow t-1$.
\STATE sample $b\leftarrow \mathcal{B}_{\exp(-1/(2\sigma^2))}$.
\RETURN $k$ \textbf{if} $b$ is false, otherwise set $k\leftarrow k+1$, $t\leftarrow 2(k-1)$ and \textbf{goto} step 4.
\end{algorithmic}
\end{algorithm}

It is not hard to see that Algorithm~\ref{samplingCenteredNew} is also to select $k\in\mathbb{Z}^{+}$ with probability $\exp(-k/(2\sigma^2))\cdot(1-\exp(-1/(2\sigma^2))$, and then to accept $k$ with probability $\exp(-\frac{1}{2\sigma^2}k(k-1))$. Thus, the returned value of $k$ has the desired (relative) probability density $$\exp(-k/(2\sigma^2))\cdot\exp(-\frac{1}{2\sigma^2}k(k-1))=\exp(-k^2/(2\sigma^2)),$$
which guarantees the correctness of Algorithm~\ref{samplingCenteredNew}. In particular, $b$ is false in step 1 or step 2 is equivalent to $k(k-1)=0$. Thus, Algorithm~\ref{samplingCenteredNew} directly outputs $k=0$ or $k=1$ in these two cases.

The difference between Algorithm~\ref{samplingCentered} and Algorithm~\ref{samplingCenteredNew} is that the latter decides to accept or reject $k$ at once as long as $k\geq2$, while the former determines the final value of $k$ firstly, and then decides to accept or reject it. Specifically, Algorithm~\ref{samplingCentered} needs $(k+1)$ Bernoulli random values from $\mathcal{B}_{\exp(-1/(2\sigma^2))}$ to determines the value of $k$ in step 1. However, the last Bernoulli value will be wasted if $k$ is rejected in step 2. So, the basic idea of Algorithm~\ref{samplingCenteredNew} is to accept or reject $k$ earlier (before the final value of $k$ is determined), and then decides whether or not to return $k$ by using one more Bernoulli random value from $\mathcal{B}_{\exp(-1/(2\sigma^2))}$, or tries to find a larger $k$ if necessary. In other words, a Bernoulli random value will be saved in Algorithm~\ref{samplingCenteredNew} if the selected $k$ is rejected. Therefore, the expected number of uniform deviates used by Algorithm~\ref{samplingCenteredNew} will be smaller as compared to Algorithm~\ref{samplingCentered}. The actual number for a given $\sigma$ can be estimated as follows.

\begin{proposition}
The expected number of Bernoulli random values from $\mathcal{B}_{1/\sqrt{e}}$ used by Algorithm~\ref{samplingCenteredNew} with $\sigma=1$ for sampling from $\mathcal{D}_{\mathbb{Z}^+,1}$ is about $3.684$. \label{ComplexityStdNew}
\end{proposition}
\begin{proof}
Let $p_1=\exp(-1/(2\sigma^2))=\exp(-1/2)$ and $p_0=1-p_1$. The expected number of Bernoulli random values from $\mathcal{B}_{p_1}$ used by Algorithm~\ref{samplingCenteredNew} for outputting some integer $k\geq0$ is
$$\sum_{k=0}^\infty(1+k+k(k-1))p_1^kp_1^{k(k-1)}p_0=\sum_{k=0}^\infty(1+k^2)p_1^{k^2}p_0.$$
The expected number of Bernoulli random values from $\mathcal{B}_{p_1}$ used by Algorithm~\ref{samplingCenteredNew} for rejecting some integer $k\geq2$ and going back step 1 is
\begin{gather*}
p_1p_1(3\cdot p_0+4\cdot p_1p_0)+\ldots+p_1^2p_1^2p_1(6\cdot p_0+7p_1p_0+8\cdot p_1^2p_0+9\cdot p_1^3p_0)+\ldots\\
=\sum_{k=2}^\infty p_1^{k-1}p_1^{(k-1)(k-2)}p_1\left(\sum_{j=1}^{2(k-1)}(1+(k-1)^2+j)\cdot p_1^{j-1}p_0\right).
\end{gather*}
In particular, if $k=0$ or $k=1$, then $k$ will directly be accepted in step 1 or step 2. Thus, if Algorithm~\ref{samplingCenteredNew} outputs some integer $k$, or if it rejects some integer $k$ before going back to step 1, the expected number of Bernoulli random values from $\mathcal{B}_{p_1}$ used can be given by
$$\sum_{k=0}^\infty(1+k^2)p_1^{k^2}p_0+\sum_{k=2}^\infty p_1^{k-1}p_1^{(k-1)(k-2)}p_1\left(\sum_{j=1}^{2(k-1)}(1+(k-1)^2+j)\cdot p_1^{j-1}p_0\right),$$
which is appropriately equal to $2.54149$. The average number of Algorithm~\ref{samplingCenteredNew} going back to step 1 until it outputs an integer is $1/\sum_{k=0}^\infty p_1^kp_0\approx1/0.689875$. Therefore, the expected number of Bernoulli random values from $\mathcal{B}_{p_1}$ used by Algorithm~\ref{samplingCenteredNew} for successfully generating an integer from $\mathcal{D}_{\mathbb{Z}^+,1}$  is about $2.54149/0.689875\approx3.68399$.
\end{proof}

\section{The Computational Complexity of Sampling from $\mathcal{B}_{\exp\left(-\frac12x(2k+x)\right)}$}
\subsection{Karney's Sampling Algorithm for $\mathcal{B}_{\exp\left(-\frac12x(2k+x)\right)}$}
Algorithm~\ref{BernoulliExpX} involves a random selector $C(m)$ that returns $-1$, $0$ and $1$ with probability $1/m$, $1/m$ and $1-2/m$ respectively, where $m=2k+2$. It is not hard to see that the random selector $C(m)$ can be exactly implemented by sampling two Bernoulli random values from $\mathcal{B}_{2/m}$ and $\mathcal{B}_{1/2}$ respectively. Sampling a Bernoulli random value from $\mathcal{B}_{2/m}$ is equivalent to sampling one uniform deviate $u\in(0,1)$ and deciding whether $2/m>u$ or $2/m<u$. In particular, if $k=0$, i.e., $m=2$, then the random selector $C(m)=C(2)$ only uses a random bit and returns $-1$ or $0$ in this case. Therefore, for simplicity, we can say that a random selector $C(m)$ uses one uniform deviate if $m\geq4$, and in particular it is ``free of charge'' if $m=2$.

\begin{lemma}
For a given $k\geq1$, the probability that Algorithm~\ref{BernoulliExpX} restarts (goes back to step 2) $n$ times can be given by
$$\left(\frac{x}{m}+\frac{2k}{m}\right)^n\left(\frac{x^n}{n!}\right),$$
where $m=2k+2$ and $n$ is a positive integer. \label{ProbBerExp}
\end{lemma}
\begin{proof}
For a given $k\geq1$, there are two cases in which the algorithm goes to step 2: (1) $z<y$ and $f=1$ with probability $x\cdot(2k/m)$; (2) $z<y$, $f=-1$ and $r<x$ with probability $x\cdot(x/m)$. Thus, the probability that the algorithm goes to step 2 is equal to
$$x\frac{x}{m}+x\cdot\frac{2k}{m}=x\left(\frac{x}{m}+\frac{2k}{m}\right).$$
After restarting one time, the probability that the algorithm goes to step 2 once again is equal to
$$\left(x\left(\frac{x}{m}+\frac{2k}{m}\right)\right)\frac{x}{2}\cdot\frac{x}{m}+\left(x\left(\frac{x}{m}+\frac{2k}{m}\right)\right)\frac{x}{2}\cdot\frac{2k}{m}=\frac{x^2}{2}\left(\frac{x}{m}+\frac{2k}{m}\right)^2.$$
Generally, the probability that the algorithm restarts $n$ times ($n\geq1$) can be given by
$$\left(\frac{x}{m}+\frac{2k}{m}\right)^n\left(\frac{x^n}{n!}\right),$$
where ${x^n}/{n!}$ is the probability that a set of uniform deviates $z_1, z_2, \ldots z_n$ over the range $(0,1)$ satisfy $x>z_1>z_2>\ldots>z_n$. The proof can be completed by using the mathematical induction on $n$.
\end{proof}

\begin{proposition}
Let $x\in(0,1)$. For a given $k\geq1$, the expected number of uniform deviates used by Algorithm~\ref{BernoulliExpX} for sampling from $\mathcal{B}_{\exp\left(-\frac12x(2k+x)\right)}$ can be given by
$$\frac{(4k+x+3)\cdot\tau_k(x)-2k-3}{2k+x},$$
where $\tau_k(x)=\exp(x\frac{2k+x}{2k+2})$. \label{ComplexityBerExp}
\end{proposition}
\begin{proof}
We can see that Algorithm~\ref{BernoulliExpX} needs
$$1+1+\left(\frac{x}{m}\right)\cdot1=2+\left(\frac{x}{m}\right)$$
uniform deviates on average every time it restarts, where the first deviate is for step 2, the second one is used by the random selector $C(m)$, and the third one is possibly required when $f=0$ in Algorithm~\ref{BernoulliExpX}. By Lemma~\ref{ProbBerExp}, the probability of restarting $n-1$ times ($n\geq1$) is
$$\left(\frac{x}{m}+\frac{2k}{m}\right)^{n-1}\left(\frac{x^{n-1}}{(n-1)!}\right).$$
By the binomial theorem, for a given $n\geq1$, Algorithm~\ref{BernoulliExpX} uses
$$\sum_{i=0}^{n-1}\left(n-1\atop i\right)\left(\frac{x}{m}\right)^i\left(\frac{2k}{m}\right)^{n-1-i}(2n-2+i)$$
uniform deviates on average if it restarts $n-1$ times. For a given $k\geq1$, there are three cases in which the algorithm goes to step 6 and returns the result before it goes to step 2: (1) $z>y$ with probability $(1-x)$ at a cost of one new uniform deviate; (2) $f=-1$ with probability $x(1/m)$ at a cost of two new uniform deviates; (3) $f=0$ and $r>x$ with probability $x(1/m)(1-x)$ at a cost of three new uniform deviates. Then, after restarting $n-1$ times ($n\geq1$), there are three cases in which the algorithm goes to step 6 and returns the result before it goes back to step 2 once again. \\
\noindent(1) $z>y$ with probability
$$\left(\frac{x}{m}+\frac{2k}{m}\right)^{n-1}\left(\frac{x^{n-1}}{(n-1)!}-\frac{x^n}{n!}\right),$$
and at a cost of
\begin{equation}
\sum_{n=1}^{\infty}\sum_{i=0}^{n-1}\left(n-1\atop i\right)\left(\frac{x}{m}\right)^i\left(\frac{2k}{m}\right)^{n-1-i}(2n-1+i)
\left(\frac{x^{n-1}}{(n-1)!}-\frac{x^n}{n!}\right)
\end{equation}
uniform deviates, where ${x^{n-1}}/{(n-1)!}-{x^n}/{n!}$ is the probability that the length of the longest decreasing sequence $x>z_1>z_2>\ldots>z_n$ is $n$.\\
(2) $f=-1$ with probability
$$\left(\frac{x}{m}+\frac{2k}{m}\right)^{n-1}\left(\frac{x^n}{n!}\right)\left(\frac{1}{m}\right),$$
and at a cost of
\begin{equation}
\sum_{n=1}^{\infty}\sum_{i=0}^{n-1}\left(n-1\atop i\right)\left(\frac{x}{m}\right)^i\left(\frac{2k}{m}\right)^{n-1-i}(2n+i)
\left(\frac{x^n}{n!}\right)\left(\frac{1}{m}\right)
\end{equation}
uniform deviates.\\
(3) $f=0$ and $r>x$ with probability
$$\left(\frac{x}{m}+\frac{2k}{m}\right)^{n-1}\left(\frac{x^n}{n!}\right)\left(\frac{1}{m}\right)(1-x),$$
and at a cost of
\begin{equation}
\sum_{n=1}^{\infty}\sum_{i=0}^{n-1}\left(n-1\atop i\right)\left(\frac{x}{m}\right)^i\left(\frac{2k}{m}\right)^{n-1-i}(2n+1+i)
\left(\frac{x^n}{n!}\right)\left(\frac{1}{m}\right)(1-x)
\end{equation}
uniform deviates. Therefore, for a given $k\geq1$, by combining the above three cases, i.e., summing Equations (1)-(3), we can obtain the expected number of uniform deviates used by Algorithm~\ref{BernoulliExpX}, which can be reduced to a function of $x$:
$$\frac{(4k+x+3)\cdot\tau_k(x)-2k-3}{2k+x},$$
where $\tau_k(x)=\exp(x\frac{2k+x}{2k+2})$.
\end{proof}

When $k=0$, since the random selector $C(2k+2)=C(2)$ returns $-1$ or $0$ using only one random bit, Karney switches the order of step 2 and step 3 in Algorithm~\ref{BernoulliExpX} to save the uniform deviate consumption. In this case, the probability that Algorithm~\ref{BernoulliExpX} restarts $n$ times is
$$\left(\frac{1}{2}\right)^n\left(\frac{x^n}{n!}\right)x^n,$$
where $x^n$ is the probability that a set of uniform deviates $r_1, r_2, \ldots r_n$ over the range $(0,1)$ satisfy $r_i<x$ for $1\leq i \leq n$. Following the idea of Proposition~\ref{ComplexityBerExp} we can also give an estimate for Algorithm~\ref{BernoulliExpX} with $k=0$.
\begin{proposition}
Let $x\in(0,1)$. when $k=0$, the expected number of uniform deviates used by Algorithm~\ref{BernoulliExpX} (with switching the order of step 2 and step 3) for sampling from $\mathcal{B}_{\exp\left(-\frac12x(2k+x)\right)}=\mathcal{B}_{\exp\left(-\frac{x^2}{2}\right)}$ can be given by
$$\frac{(x+2)\cdot\tau_0(x)-2}{2x},$$
where $\tau_0(x)=\exp(\frac{x^2}{2})$. \label{ComplexityBerExp0}
\end{proposition}
\begin{proof}
If $k=0$, after restarting $n-1$ times ($n\geq1$), there are three cases in which Algorithm~\ref{BernoulliExpX} goes to step 6 and returns the result before it goes back to step 2 again.\\
\noindent(1) $f=-1$ with probability
$$\left(\frac{1}{2}\right)^n\left(\frac{x^{n-1}}{(n-1)!}\right)x^{n-1},$$
and at a cost of
\begin{equation}
\sum_{n=1}^{\infty}\left(\frac{1}{2}\right)^n\left(\frac{x^{n-1}}{(n-1)!}\right)x^{n-1}(2n-2)
\end{equation}
uniform deviates.\\
(2) $z>y$ with probability
$$\left(\frac{1}{2}\right)^nx^{n-1}\left(\frac{x^{n-1}}{(n-1)!}-\frac{x^n}{n!}\right),$$
and at a cost of
\begin{equation}
\sum_{n=1}^{\infty}\left(\frac{1}{2}\right)^nx^{n-1}\left(\frac{x^{n-1}}{(n-1)!}-\frac{x^n}{n!}\right)(2n-1)
\end{equation}
uniform deviates.\\
(3) $f=0$ and $r>x$ with probability
$$\left(\frac{1}{2}\right)^nx^{n-1}\left(\frac{x^n}{n!}\right)(1-x),$$
and at a cost of
\begin{equation}
\sum_{n=1}^{\infty}\left(\frac{1}{2}\right)^nx^{n-1}\left(\frac{x^n}{n!}\right)(1-x)(2n)
\end{equation}
uniform deviates. Therefore, when $k=0$, by summing Equations (4)-(6), we have the expected number of uniform deviates used by Algorithm~\ref{BernoulliExpX}, which can be reduced to a function of $x$:
$$\frac{(x+2)\cdot\tau_0(x)-2}{2x},$$
where $\tau_0(x)=\exp(\frac{x^2}{2})$.
\end{proof}

Let's go back to Algorithm~\ref{KarneyStdN}. For a given $k\geq0$, the average number of step 4 invoking Algorithm~\ref{BernoulliExpX}, denoted by $\mathbf{t}_k(x)$, is equal to
$$\sum_{i=1}^{k}i\left(\exp\left(-x\frac{2k+x}{2k+2}\right)\right)^{i-1}\left(1-\exp\left(-x\frac{2k+x}{2k+2}\right)\right)
+(k+1)\left(\exp\left(-x\frac{2k+x}{2k+2}\right)\right)^k.$$
In particular, when $k=0$, we have $\mathbf{t}_0(x)=1$. The value of $k$ conforms to the discrete Gaussian distribution $\mathcal{D}_{\mathbb{Z}^+,1}$, and $x$ is a uniformly random number over the range $(0,1)$. Then, by Proposition~\ref{ComplexityBerExp} and Proposition~\ref{ComplexityBerExp0}, the expected number of uniform deviates used by Algorithm~\ref{KarneyStdN} for sampling from $\mathcal{B}_{\exp\left(-\frac12x(2k+x)\right)}$ in step 4 can be given by
$$\mathcal{D}_{\mathbb{Z}^+,1}(0)\int_{0}^{1}\frac{(x+2)\cdot\tau_0(x)-2}{2x}dx+
\sum_{k=1}^{\infty}\mathcal{D}_{\mathbb{Z}^+,1}(k)\left(\int_{0}^{1}\frac{(4k+x+3)\cdot\tau_k(x)-2k-3}{2k+x}\mathbf{t}_k(x)dx\right),$$
where $\tau_0(x)$ and $\tau_k(x)$ with $k\geq1$ are defined as in Proposition~\ref{ComplexityBerExp0} and Proposition~\ref{ComplexityBerExp} respectively. By performing numerical calculation, we can verify that the actual value calculated by this formula is about $2.19414$.

\subsection{An Alternative Sampling Algorithm for $\mathcal{B}_{\exp\left(-\frac12x(2k+x)\right)}$}
In this subsection, we show that one may not need to use Algorithm~\ref{BernoulliExpX} to sample from the Bernoulli distribution $\mathcal{B}_{\exp\left(-\frac12x(2k+x)\right)}$ with an integer $k\geq0$ and a uniformly random number $x\in(0,1)$. Instead, sampling from $\mathcal{B}_{\exp\left(-\frac12x(2k+x)\right)}$ can be simply decomposed into two sampling procedures: sampling from $\mathcal{B}_{e^{-kx}}$ and from $\mathcal{B}_{e^{-x^2/2}}$ respectively
since $\exp\left(-\frac12x(2k+x)\right)=e^{-kx}e^{-x^2/2}$.

It is clear that we can repeatedly use the generalized version of Algorithm~\ref{BernoulliExph} (replacing $1/2$ with $x$) at most $k$ times to sample from $\mathcal{B}_{e^{-kx}}$. When sampling from $\mathcal{B}_{e^{-x^2/2}}$, however, it is hard for us to directly compare a uniform deviate with the value of $x^2/2$ for a given real number $x$. So, we use the following algorithm, namely Algorithm~\ref{BernoulliExpxy}, to address this problem.

\begin{algorithm}
\caption{Generating a Bernoulli random value which is true with probability $e^{-xy}$}\label{BernoulliExpxy}
\begin{algorithmic}[1]
\REQUIRE $x,y\in(0,1)$
\ENSURE a Bernoulli random value from $\mathcal{B}_{e^{-xy}}$
\STATE set $w\leftarrow x$, $n\leftarrow 0$.
\STATE sample a uniform deviate $u\in(0,1)$, and \textbf{goto} step 5 unless $u<w$.
\STATE sample a uniform deviate $v\in(0,1)$, and \textbf{goto} step 5 unless $v<y$.
\STATE set $w\leftarrow u$, $n\leftarrow n+1$ and \textbf{goto} step 2.
\STATE \textbf{return} \textbf{true} if $n$ is even, otherwise \textbf{return} \textbf{false}.
\end{algorithmic}
\end{algorithm}

In fact, Algorithm~\ref{BernoulliExpxy} can be viewed as a special version of Algorithm~\ref{BernoulliExpX}, as the basic idea behind it follows from Algorithm~\ref{BernoulliExpX}. It samples two sets of uniform deviates $u_1, u_2, \ldots$ and $v_1, v_2, \ldots$ over the range $(0,1)$, and then determines the maximum value $n\geq 0$ such that $x>u_1>u_2>\ldots>u_n$ and $v_i<y$. If $n$ is even, it returns \textbf{true}, and the probability is exactly equal to
$$\left(1-xy\right)+\left(\frac{x^2}{2!}y^2-\frac{x^3}{3!}y^3\right)+\ldots=\sum_{n=0}^\infty\left(\frac{x^{n}}{n!}y^{n}-\frac{x^{n+1}}{(n+1)!}y^{n+1}\right)=e^{-xy}.$$

\begin{proposition}
Let $x, y\in(0,1)$. The expected number of uniform deviates used by Algorithm~\ref{BernoulliExpxy} for sampling from $\mathcal{B}_{e^{-xy}}$ can be given by $(e^{xy}(1+y)-1)/{y}$.
\label{ComplexityExpxy}.
\end{proposition}
\begin{proof}
For a given $n\geq1$, the probability that Algorithm~\ref{BernoulliExpxy} restarts (goes back to step 2) $(n-1)$ times is
$$\frac{x^{n-1}}{(n-1)!}y^{n-1}.$$
Then, after restarting $(n-1)$ times, there are two cases in which Algorithm~\ref{BernoulliExpxy} goes to step 5 and returns the result before it goes back to step 2 again.\\
\noindent(1) $u>w$ with probability
$$\left(\frac{x^{n-1}}{(n-1)!}-\frac{x^{n}}{n!}\right)y^{n-1}.$$
(2) $u<w$ and $v>y$ with probability
$$\left(\frac{x^{n}}{n!}\right)y^{n-1}(1-y).$$
Algorithm~\ref{BernoulliExpxy} uses $(2n-1)$ uniform deviates in the first case, while it uses $2n$ uniform deviates in the second case. Therefore, the expected number of uniform deviates can be given by
$$\sum_{n=1}^{\infty}\left(\frac{x^{n-1}}{(n-1)!}y^{n-1}(2n-1)+\left(\frac{x^{n}}{n!}\right)y^{n-1}(1-y)(2n)\right),$$
which can be reduced to $(e^{xy}(1+y)-1)/{y}$.
\end{proof}

For a given $k\geq0$ and a uniformly random number $x\in(0,1)$, applying Proposition~\ref{ComplexityExpxy} gives the expected number of uniform deviates used by Algorithm~\ref{BernoulliExpxy} for sampling from $\mathcal{B}_{e^{-x^2/2}}$:
$$\int_{0}^{1}\frac{(1+x)e^{\frac{x^2}{2}}-1}{x}dx.$$
Here, we set $y\leftarrow x$ and $x\leftarrow x/2$ in Algorithm~\ref{BernoulliExpxy}, since we have the fact that
$$\frac{e^{xy}(1+y)-1}{y}<\frac{e^{xy}(1+x)-1}{x}$$
if and only if $x<y$.

\subsection{An Improved Sampling Algorithm for Standard Normal Distribution}
Combining the idea presented in Section 4.2, namely Algorithm~\ref{BernoulliExpxy}, we give Algorithm~\ref{ImprovedKarneyStdN}, which is to sample from the standard normal distribution $\mathcal{N}(0,1)$.

\begin{algorithm}
\caption{Sampling from the standard normal distribution $\mathcal{N}(0,1)$}\label{ImprovedKarneyStdN}
\begin{algorithmic}[1]
\ENSURE a sample value that conforms to $\mathcal{N}(0,1)$
\STATE sample $k\in\mathbb{Z}^{+}$ from $\mathcal{D}_{\mathbb{Z}^+,1}$ either \\
(a) by using steps 1 and 2 from Algorithm~\ref{KarneyStdN} or (b) by applying Algorithm~\ref{samplingCenteredNew} 
\STATE sample a uniformly random number $x\in(0,1)$.
\STATE sample a Boolean random values from $\mathcal{B}_{e^{-kx}}$, and \textbf{goto} step 1 if it is \textbf{false}.
\STATE sample a Boolean random values from $\mathcal{B}_{e^{-x^2/2}}$, and \textbf{goto} step 1 if it is \textbf{false}.
\STATE set $s\leftarrow \pm1$ with equal probabilities and \textbf{return} $s(k+x)$.
\end{algorithmic}
\end{algorithm}

The correctness of Algorithm~\ref{ImprovedKarneyStdN} follows from Karney's algorithm for sampling from standard normal distribution. According to Proposition~\ref{ComplexityStd} and Proposition~\ref{ComplexityStdNew}, it is better to use Algorithm~\ref{samplingCenteredNew} with $\sigma=1$ to sample $k$ from $\mathcal{D}_{\mathbb{Z}^+,1}$, since Algorithm~\ref{samplingCenteredNew} uses a smaller number of uniform deviates on average. In this subsection, for sampling from $\mathcal{B}_{\exp\left(-\frac12x(2k+x)\right)}$, we show that the expected number of uniform deviates used by step 3 and step 4 in Algorithm~\ref{ImprovedKarneyStdN} is also slightly smaller as compared to using Algorithm~\ref{BernoulliExpX}.

When sampling from $\mathcal{B}_{e^{-kx}}$ in step 3, the average number of invoking the generalized version of Algorithm~\ref{BernoulliExph}, denoted by $\mathbf{t}_k(x)$, is equal to
$$\sum_{i=1}^{k-1}i\left(e^{-x}\right)^{i-1}\left(1-e^{-x}\right)+k\left(e^{-x}\right)^{k-1}.$$
In particular, when $k=0$, we have $\mathbf{t}_0(x)=0$.

The value of $k$ conforms to the discrete Gaussian distribution $\mathcal{D}_{\mathbb{Z}^+,1}$, and $x$ is a uniformly random number from $(0,1)$. Furthermore, for a given $k\geq0$, the algorithm executes step 4 with probability $(e^{-x})^k$. Therefore, according to Corollary~\ref{ComplexityBernoulliExpx} and Proposition~\ref{ComplexityExpxy}, the expected number of uniform deviates used by Algorithm~\ref{ImprovedKarneyStdN} for sampling from $\mathcal{B}_{\exp\left(-\frac12x(2k+x)\right)}$ in step 3 and step 4 can be given by
$$\sum_{k=1}^{\infty}\mathcal{D}_{\mathbb{Z}^+,1}(k)\left(\int_{0}^{1}\mathbf{t}_k(x)e^{x}dx\right)
+\sum_{k=0}^{\infty}\mathcal{D}_{\mathbb{Z}^+,1}(k)\left(\int_{0}^{1}\frac{(1+x)e^{\frac{x^2}{2}}-1}{x}(e^{-x})^kdx\right)\approx2.01799,$$
which is slightly smaller than $2.19414$, i.e., the expected number of uniform deviates used by Algorithm~\ref{KarneyStdN} for sampling from $\mathcal{B}_{\exp\left(-\frac12x(2k+x)\right)}$.

\section{Experimental Results}
On a laptop computer (Intel i7-8550U, 16GB RAM, the g++ compiler and enabling -O3 optimization option), and following the implementation of Algorithm~\ref{KarneyStdN} in Karney's C++ library RandomLib, we implemented our proposed algorithms. The implementation of Algorithm~\ref{samplingCenteredNew} and Algorithm~\ref{ImprovedKarneyStdN} is simply based on the adaption of ExactNormal.hpp as well as the runtime environment provided by RandomLib \footnote{`RandomLib' is available at http://randomlib.sourceforge.net/.}. 

We tested the average numbers of Bernoulli random values consumed by Algorithm~\ref{KarneyStdN} and Algorithm~\ref{ImprovedKarneyStdN} for outputting an integer from $\mathcal{D}_{\mathbb{Z}^+,1}$. The average quantities measured in practice are consistent with the expected values estimated in Proposition \ref{ComplexityStd} and Proposition \ref{ComplexityStdNew}. We also verified that the average numbers of uniform deviates used by Algorithm~\ref{KarneyStdN} and Algorithm~\ref{ImprovedKarneyStdN} in practice for sampling from $\mathcal{B}_{\exp\left(-\frac12x(2k+x)\right)}$ is about $2.19$ and $2.02$ respectively.

A uniform deviate could be made up of one or several digits, and it will only be treated as one deviate no matter how many digits it actually uses. The comparison of two deviates is realized digit-by-digit, and each digit of a deviate is generated online according to the actual needs. In fact, one digit could consist of only one bit or a small number of bits, such as $4$ bits, $8$ bits or $16$ bits. We call the number of bits in each digit the \emph{digit size}, which can be specified through assigning an integer value to the variable \texttt{bits} in the source code. If uniformly random bits can be produced in batches efficiently (say a typical software approach), then a larger \texttt{bits} could bring better actual sampling performance.

\begin{table}[bhtp]
\centering
\caption{The performance of Algorithm~\ref{KarneyStdN} and Algorithm~\ref{ImprovedKarneyStdN} ($10^6$ sample values per second)}\label{comparison}
\begin{tabular}{lcccc}
\hline
                                            & \texttt{bits}$=1$ & \texttt{bits}$=4$ & \texttt{bits}$=8$ & \texttt{bits}$=16$ \\
\hline
Alg.~\ref{KarneyStdN} \cite{Karney2016}     & $2.109$ & $3.248$ & $3.470$ & $3.502$  \\
Alg.~\ref{ImprovedKarneyStdN}(a)            & $2.159$ & $3.298$ & $3.503$ & $3.561$  \\
Alg.~\ref{ImprovedKarneyStdN}(b)            & $2.532$ & $3.929$ & $4.017$ & $4.135$  \\
\hline
\end{tabular}
\end{table}

Using Algorithm~\ref{samplingCenteredNew} with $\sigma=1$, one can get about $26.5\times10^6$ sample values per second from the discrete Gaussian distribution $\mathcal{D}_{\mathbb{Z},1}$, while using Algorithm~\ref{samplingCentered} with $\sigma=1$ one can obtain only about $21.5\times10^6$ sample values per second from $\mathcal{D}_{\mathbb{Z}^+,1}$. Table~\ref{comparison} shows the performance of Algorithm~\ref{KarneyStdN} and Algorithm~\ref{ImprovedKarneyStdN}. One can see that the performance advantage of Algorithm~\ref{ImprovedKarneyStdN} can be shown even without using Algorithm~\ref{samplingCenteredNew}. This could be well explained by the fact that the expected number of uniform deviates used by Algorithm~\ref{ImprovedKarneyStdN} for sampling from $\mathcal{B}_{\exp\left(-\frac12x(2k+x)\right)}$ is slightly smaller as compared to using Algorithm~\ref{BernoulliExpX} in Algorithm~\ref{KarneyStdN}.

\section{Conclusion and Future Work}
In this paper, under the random deviate model, we discuss the computational complexity of Karney's exact sampling algorithm for the standard normal distribution. We present an improved algorithm for sampling $\mathcal{D}_{\mathbb{Z}^+,1}$, and an alternative sampling method for the Bernoulli distribution $\mathcal{B}_{\exp\left(-\frac12x(2k+x)\right)}$ with an integer $k\in\mathbb{Z}^+$ and a real number $x\in(0,1)$. The expected number of uniform deviates they used is smaller as compared to Karney's algorithms. Then, we can obtain an improved exact sampling algorithm with lower uniform deviate consumption for the standard normal distribution.

Furthermore, a discrete Gaussian distribution can also be defined over the set of all the integers $\mathbb{Z}$, namely $\mathcal{D}_{\mathbb{Z},\sigma,\mu}(k)={\rho_{\sigma,\mu}(k)}/{\rho_{\sigma,\mu}(\mathbb{Z})}$ for all $k\in\mathbb{Z}$, where $\rho_{\sigma,\mu}(\mathbb{Z})=\sum_{x\in\mathbb{Z}}\rho_{\sigma,\mu}(x)$. In recent years, the issue of sampling from discrete Gaussian distributions over the integers has received increasing attention because of its application in cryptography \cite{Gentry2008,Dwarakanath2014,Micciancio2017,zhao2020}. The methods of sampling from a continuous Gaussian distribution are not trivially applicable for the discrete case. As a discretization version of the algorithm for sampling exactly from the standard normal distribution, Karney also presented an exact sampling algorithm for the discrete Gaussian distribution over the integers $\mathbb{Z}$, whose parameters ($\sigma$ and $\mu$) are rational numbers. (see Algorithm D in \cite{Karney2016}.) Sampling from $\mathcal{D}_{\mathbb{Z}^+,1}$ and sampling from $\mathcal{B}_{\exp\left(-\frac12x(2k+x)\right)}$ are also two key subroutines in Karney's sampling algorithm for discrete Gaussian distributions. The computational complexity and the actual performance of this algorithm could also be improved by our work in this paper.

Our complexity analysis of Karney's algorithms is based on estimating expected number of uniform deviates used, namely the random deviate model, which is still a ``coarse-grained'' complexity model as compared to the random bit model, although it can be used to well explain the complexity advantage of our improved algorithms. It would be also interesting to investigate Karney's algorithms under the random bit model, in which the complexity is measured by the expected number of random bits used by the sampling algorithm. In our opinion, following the argument presented in this work, the complexity of Karney's algorithms under the random bit model could be well estimated if a theoretical estimate of the expected number of random bits used by Algorithm~\ref{BernoulliExph} as well as its generalized version can be given.

\bibliographystyle{plain}
\bibliography{mybibfile}

\end{document}